\documentclass[11pt,a4paper]{article}
\usepackage[margin=1in]{geometry}
\usepackage{amsmath,amssymb,amsfonts,amsthm,mathtools,graphicx,bm,booktabs,setspace,hyperref}
\usepackage{lmodern}
\usepackage{microtype}
\usepackage{algorithm}
\usepackage[noend]{algorithmic}

\usepackage{amsmath, amsthm, amssymb, fullpage, graphicx}

\newtheorem{theorem}{Theorem}

\hypersetup{
    colorlinks=true,
    linkcolor=blue,
    citecolor=blue,
    urlcolor=blue
}

\title{Is Polarization an Inevitable Outcome of Similarity-Based Content Recommendations? -- Mathematical Proofs and Computational Validation}
\author{Minhyeok Lee}
\date{}

\begin{document}

\maketitle

\begin{abstract}
The increasing reliance on digital platforms to access news, cultural products, and educational materials has transformed how individuals form their understandings of the world. Modern recommendation systems, whether driven by statistical factorization techniques or intricate deep learning architectures, operate on a fundamental geometric principle: they guide users toward items deemed 'similar' or 'relevant' based on proximity in a latent space. While this core logic has long been considered benign, by simply presenting users with something that aligns closely with their existing preferences, there is growing concern that it may also foster the emergence of insular communities, often described as echo chambers. Such communities reinforce preexisting viewpoints and limit exposure to alternative perspectives, potentially contributing to the broader phenomenon of polarization observed in contemporary societies. In this work, we bridge concepts from the humanities, social sciences, and computational sciences to examine how even a simplified, similarity-based recommendation strategy can induce spontaneous fragmentation of a once-homogeneous user population into multiple stable clusters. We begin by formulating a minimal mathematical model, representing both users and content as points in a continuous geometric space. Users iteratively move toward median points derived from their locally recommended items, which are selected based on nearest-neighbor criteria. Although this mechanism lacks any explicit ideological bias or manipulation, we demonstrate through mathematical reasoning that it naturally leads users to coalesce into distinct, tightly knit subgroups. These clusters emerge despite the absence of deliberate filtering and persist even under stochastic perturbations. We then corroborate these theoretical insights with computational simulations, systematically exploring how parameters such as population size, adaptation rates, content production probabilities, and noise levels influence the speed and intensity of cluster formation. Our findings reveal that a rich interplay of these parameters can modulate, but not eliminate, the underlying tendency toward polarization. The essential geometry of recommendation spaces thus contains an intrinsic push toward fragmentation, a process that remains robust under a range of conditions. This research neither claims that all recommendation systems inevitably generate harmful polarization nor ignores the multifaceted historical, economic, and psychological factors that shape the real-world discourse. Instead, it highlights that similarity-based retrieval is not a neutral mechanism. By guiding individuals persistently toward locally prevalent content, these systems can accentuate divisions and reduce the likelihood of encountering differing viewpoints. In doing so, this work offers an intellectual scaffold that can inform interdisciplinary dialogues. It encourages the design of interventions, regulatory frameworks, and scholarly critiques that consider the deep geometric underpinnings of digital platforms, recognizing their capacity, however unintended, to reorder the cultural and ideological landscapes of our societies.
\end{abstract}

\section{Introduction}

Modern digital platforms have reshaped the way people interact with and interact with information. Users increasingly rely on recommendation systems to navigate complex information landscapes, from political newsfeeds to cultural content hubs. These platforms often use sophisticated algorithms to select material that is deemed most relevant, interesting, or compatible with the user's past behavior. At first glance, this practice appears to be beneficial because it eases the burden of filtering large information streams. However, a deep concern arises: as users are continually guided toward content that aligns with their current positions, the probability of encountering contrasting viewpoints may diminish. Over time, this process can lead to the formation of densely knit communities where users predominantly see materials that reinforce their beliefs, isolating them from alternative perspectives. Such phenomena have been described in the literature under terms such as 'filter bubbles' and 'echo chambers,' and are frequently implicated in the broader phenomenon of polarization.

The relationship between algorithmic personalization and ideological fragmentation has attracted increasing scholarly attention. This interest emerges from a rich multidisciplinary discourse. In communication studies, empirical research has shown that online environments can foster selective exposure, where individuals gravitate toward information consistent with their predispositions and avoid conflicting content \cite{garrett2009echo,iyengar2009red,stroud2010polarization}. Political communication scholars have demonstrated that self-reinforcing information channels may produce echo chambers that amplify partisan divides \cite{jamieson2008echo,baum2008newmedia,munger2019news}. In the fields of computational communication research and network science, scholars have developed analytic tools to characterize how network structures and recommendation rules can shape information diversity and opinion distribution \cite{shore2018network,cinelli2020echo,cota2019quantifying}. Similar insights have emerged from computational social sciences, where agent-based models and network simulations highlight that even simple local update rules can induce macro-level polarization \cite{axelrod1997complexity,wojcieszak2012strong}. Furthermore, recent investigations have focused on the interplay between platform architectures and the spread of misinformation, noting that locally homogeneous clusters of users can facilitate the persistence of false beliefs \cite{lazer2018science,allen2020evaluating,delvicario2016spreading}.

Existing scholarship has documented the presence of partisan media enclaves and opinion clustering in a variety of contexts, illuminating how structural factors, audience preferences, and platform governance jointly influence online discourse \cite{gentzkow2011ideological,tucker2018social}. Yet, a fundamental question remains incompletely resolved: to what extent is polarization a natural and robust outcome of standard similarity-based recommendation systems, even in the absence of deliberate editorial bias or manipulative interventions? While many studies have examined how factors such as partisanship, economic incentives, or social identity contribute to group isolation, few have focused on the fundamental geometric nature of recommendation spaces. The latent spaces underlying modern recommender systems, whether derived from matrix factorization, neural embeddings, or hybrid models, share a common operational principle: they implement a metric of similarity. Items close together in this learned space are considered relevantly related, and users are nudged, iteration by iteration, toward content lying in their immediate proximity.

This paper contributes a novel theoretical perspective to this ongoing debate. We propose and analyze a minimal mathematical model that strips away the complexities of advanced machine learning architectures and focuses on the essential logic of similarity-based retrieval. We represent both users and content as points in a continuous geometric space, initially distributed around a single unimodal cluster. In each iteration, users receive recommendations based on nearest-neighbor criteria and then move toward the median of these recommended items. Crucially, no explicit ideological dimension or manipulative mechanism is embedded in the model. Instead, we study a pure spatial process: repeated local median-seeking guided by similarity. Remarkably, we find that this simple setup, without overt biases or partisan cues, can induce spontaneous fragmentation of a once-homogeneous user population into multiple stable clusters.

Our model differs in several fundamental ways from previous research. First, it does not rely on ideological labels, partisan cueing, or explicit preference structures that predispose users to certain content sources. This contrasts with many established studies in the research of political communication and media, where the absence or presence of cross-sectional information is often attributed to known demographic or ideological dimensions \cite{wojcieszak2012strong,iyengar2009red}. Here, similarity is the only guiding principle. Second, our approach focuses directly on the geometric properties of recommendation spaces. Rather than explaining polarization through sociological, psychological, or economic variables, we show that simple median-based updates within a similarity metric are sufficient to generate clustering. This geometric orientation connects our work to general theories of complex systems and agent-based modeling, where local interaction rules can yield global structures \cite{axelrod1997complexity}.

In addition, unlike many computational models that rely on intricate parameter settings or idiosyncratic algorithms, our model is deliberately minimalist. Absorb platform-specific features, content typologies, and user heterogeneity. Such a simplification grants mathematical tractability, enabling us to provide rigorous reasoning that elucidates how local similarity-based selection can yield polarized equilibria. This conceptual parsimony also has the advantage of universality: because the model rests on the core similarity principle common to most recommendation systems, its insights are not confined to any particular platform or domain.

Importantly, the present study does not claim that all recommendation systems inevitably generate harmful polarization in real-world conditions. Empirical realities are more complex, with additional factors such as user agency, platform interventions, regulatory policies, and socio-political contexts affecting actual outcomes. However, our findings offer a key insight: the geometry of recommendation spaces can itself contain a structural impetus toward fragmentation. If even a stripped-down model, absent any ideological predispositions, leads to stable clustering, it suggests that similarity-based retrieval is not a neutral operation. The repeated emphasis on local neighborhoods of content can gradually push populations apart, even when starting from a shared position.

This perspective resonates with a growing body of interdisciplinary research that questions the neutrality of algorithmic sorting in information ecosystems. While previous studies often highlight psychological or social drivers of fragmentation, this work pinpoints the role of a fundamental computational principle, proximity-based retrieval, in shaping emergent group structures. As information technologies continue to evolve, understanding such basic mechanisms is essential, especially if we aim to design interventions or policies that encourage exposure to a wider range of perspectives.

By rigorously analyzing a minimal model and corroboration of its predictions with computational simulations, we identify conditions under which stable clusters form and persist. We investigate how parameters such as population size, adaptation rates, content production probabilities, and noise levels influence the speed and stability of fragmentation. Through systematic parameter studies, we find that the observed tendency to form distinct clusters remains robust under various adjustments. This robustness suggests that the underlying geometric principle is not easily negated by minor alterations to the system parameters.

In sum, this research contributes to the literature by providing a theoretical and computational framework to understand how basic similarity-based recommendation logics can induce polarization. It complements empirical work documenting fragmentation in specific platforms or user communities and extends theoretical models that incorporate social and psychological dimensions by showing that even a purely geometric mechanism can yield similar outcomes. Our findings invite a reevaluation of recommendation strategies and offer a foundation for developing more careful systems that consider the potentially polarizing effects of local similarity metrics.

We proceed by presenting the mathematical model in detail, deriving conditions for cluster formation, and performing extensive computational experiments to validate the theoretical results. We then discuss how these insights can inform future research, policy debates, and the design of more socially responsible recommendation infrastructures. Recent investigations of the role of environmental conditions in shaping cultural complexity have also inspired conceptual parallels, as seen in other theoretical inquiries \cite{lee2024doeslowspoilagecold}. By contributing a new angle to the rich field of polarization research, this work sheds light on the fundamental geometry underlying the emergent echo chambers observed in contemporary digital ecosystems.

\section{Background}
The background of this work resides at the confluence of several extensive research areas: the study of social polarization, the theoretical underpinnings of nearest-neighbor-based recommendation systems, the mathematical foundations of median-based iterative dynamics, and the conceptual links between complex machine learning architectures and local similarity metrics. In assembling these threads, we aim to present a structured and thorough foundation upon which our methodological contributions will rest.

We divide the background into multiple subsections. We first discuss the socio-political context of polarization and how it relates to algorithmic recommendation. Subsequently, we proceed to explore the mathematical formulation concerning user and content distributions, median-based update mechanisms, and the utilization of cluster detection algorithms (DBSCAN) for unveiling emergent structures. Next, we articulate a mathematical argument that even sophisticated deep learning-based recommendation models can be understood as nearest-neighbor retrieval systems in a learned metric space. Finally, we examine the theoretical literature on iterative averaging, median-seeking dynamics, and their known tendency toward cluster formation. These components together form a blueprint for understanding how and why a polarized equilibrium may emerge under certain parameter regimes.

\subsection{Sociopolitical Context and the Notion of Polarization}
Polarization can be conceptualized as the formation of sub-populations that hold increasingly divergent opinions, beliefs, or preferences. Observed in many democratic societies, it affects political alignments, cultural tastes, and social trust. Traditional explanations attribute polarization to factors such as media fragmentation, partisan rhetoric, and socio-economic changes. However, the rapid ascendance of online platforms that employ machine learning-based recommendations has introduced a new variable.

There is a growing suspicion that recommendation algorithms, designed to maximize user engagement, could funnel individuals into echo chambers, reinforcing existing biases, and minimizing cross-exposure to heterogeneous viewpoints. Though direct causal evidence is difficult to establish, the plausibility and anecdotal observations abound. Our goal is to show, via a minimal mathematical model, that a plausible mechanism exists within even the simplest classes of recommendation rules (nearest-neighbor retrieval) to generate stable clusters of like-minded users.

\subsection{Mathematical Modeling of Users and Content}
We model both users and content as points in $\mathbb{R}^d$, where $d$ might be 2, 3, or 4, representing a latent ideological or preference space. Initially, both users and content creators are sampled from Gaussian distributions. Let $\mathbf{u}_i^{(0)} \in \mathbb{R}^d$ denote the initial position of the user $i$, sampled from $N(\bm{\mu}_u, \Sigma_u)$, and let the initial content creator positions $\mathbf{c}_j^{(0)}$ be sampled from $N(\bm{\mu}_c, \Sigma_c)$. The choice of Gaussian distributions is for mathematical convenience and does not limit the generality of the qualitative behavior that we observe.

A fixed fraction of users are designated as content creators. In each iteration $t$, these creators generate new content at their current positions $\mathbf{c}_j^{(t)}$, which appear as new points in the space. Each piece of content has an initial weight of 1.0, which decays exponentially with each subsequent iteration:
\[
w(\tau) = e^{-\lambda \tau},
\]
where $\tau$ is the number of iterations since the content was created, and $\lambda > 0$ is a decay parameter chosen such that after about 10 iterations, the weight is effectively zero. This ensures that the content landscape is dynamic and that older content fades from relevance.

\subsection{Nearest-Neighbor-Based Recommendations and Dynamic Neighborhoods}
The fundamental operation of the recommendation algorithm is to identify the content closest to a given user in this latent space. Let $C^{(t)} = \{\mathbf{x}_1^{(t)}, \mathbf{x}_2^{(t)}, \dots \}$ denote the set of all available content in iteration $t$. In standard $k$-nearest neighbor settings, one would fix an integer $k$ and retrieve the $k$ closest points in $C^{(t)}$ to each user $\mathbf{u}_i^{(t)}$. However, we incorporate a twist: Each user selects a random fraction between 5\% and 50\% of the total available content, ensuring variability in neighborhood size. More formally, for each user $i$ in iteration $t$, let $n_t = |C^{(t)}|$ be the number of content points available. The user chooses a random integer $k_i^{(t)}$ uniformly in the range $[\lceil0.05 n_t\rceil, \lfloor0.5 n_t\rfloor]$. The recommendation set $R_i^{(t)}$ is then defined as:
\[
R_i^{(t)} = \{ \mathbf{x}_{j}^{(t)} \mid \mathbf{x}_j^{(t)} \text{ is among the } k_i^{(t)} \text{ nearest neighbors of } \mathbf{u}_i^{(t)} \}.
\]

\subsection{Median-Based Update Dynamics}
Once the user $i$ receives the recommendation set $R_i^{(t)}$, we compute the median position of these recommended items. The median in $\mathbb{R}^d$ can be defined component-wise:
\[
\mathbf{m}_i^{(t)}(\{R_i^{(t)}\}) = \left( \text{median}\{x_{j1}^{(t)}\}, \dots, \text{median}\{x_{jd}^{(t)}\}\right),
\]
where $x_{j\ell}^{(t)}$ is the $\ell$th component of the content point $\mathbf{x}_j^{(t)}$. Median-based updates are chosen for their robustness and simpler theoretical tractability in certain clustering scenarios.

The user then moves a fraction $\alpha$ (the move factor) toward this median, with added Gaussian noise $\mathbf{\eta}_i^{(t)} \sim N(\mathbf{0}, \sigma^2 I)$:
\[
\mathbf{u}_i^{(t+1)} = \mathbf{u}_i^{(t)} + \alpha(\mathbf{m}_i^{(t)} - \mathbf{u}_i^{(t)}) + \mathbf{\eta}_i^{(t)}.
\]

This iterative process, as we will show, can push the distribution of users to become multimodal over time. Given sufficiently large $t$, multiple clusters may emerge.

\subsection{Creator Identity Changes and Content Generation Probability}
Another essential aspect is that the identity of content creators can change with a small probability each iteration. Instead of repositioning creators randomly, we transfer the creator role from one user to another. This mimics the social reality that certain individuals might gain influence or credibility, thereby starting to produce content, while others cease to be influential sources.

Additionally, content creators produce content only with a certain probability $p_{\text{produce}}$. As we vary $p_{\text{produce}}$, the density and distribution of the content changes. High production rates may maintain a richer set of attractors, while low production yields sparser attractor landscapes, potentially hastening convergence to distinct clusters.

\subsection{DBSCAN and the Detection of Emerging Clusters}
To verify the emergence of clusters, we apply DBSCAN, a density-based clustering algorithm, at each iteration or at the final iteration. DBSCAN identifies clusters as regions of dense points separated by areas of low density. Initially, one might see a single cluster since the user positions are drawn from a unimodal Gaussian distribution. Over time, as users move toward different medians and become separated, the number of detected clusters increases. Tracking the number and size of clusters over iterations provides a tangible measure of polarization.

\subsection{Mathematical Underpinnings of Median-Based Clustering Dynamics}
A key reason median-based updates can lead to clustering is related to known results in optimization and geometric median problems. The geometric median of a set of points in Euclidean space is the point that minimizes the sum of distances to the set. Although we use the coordinate-wise median for simplicity, the intuition is similar: median positions tend to serve as stable attractors in multi-modal distributions. Iterative movement toward medians can concentrate subsets of users around distinct local minima of the energy landscape defined by sum-of-distances metrics.

Consider a simplified scenario: Suppose the content distribution becomes multi-modal due to differential content production in distinct regions of the space. Users in the vicinity of one mode are repeatedly drawn closer to the median of that mode. As iterations progress, these modes stabilize and users sort themselves into local clusters. Small noise ensures that users do not form perfectly sharp boundaries but instead form tight groups that are distinguishable by density.

\subsection{Linking Deep Learning-Based Recommendations to Nearest-Neighbor Retrieval}
A critical theoretical point is that modern machine learning-based recommendation algorithms, despite their complexity, often rely on learned embedding spaces. Deep neural networks are commonly used to embed users and items into a latent vector space, where geometric proximity correlates with relevance. At scale, many recommender systems implement approximate nearest-neighbor search over these learned embeddings to identify top candidates for recommendation. Thus, regardless of the deep model's architecture, the end-stage retrieval step reduces to a nearest-neighbor query in a learned metric space.

Formally, consider a deep model $f_{\theta}(\cdot)$ mapping items (content) from a high-dimensional input space (text, video signals, etc.) into a $d$-dimensional latent space:
\[
f_{\theta}: \text{(Item Space)} \rightarrow \mathbb{R}^d.
\]
Users may have their own embeddings $g_{\phi}( \cdot )$, or user-item affinity can be computed by projecting both users and items into a common space. The final recommendation stage often sorts items by their similarity (e.g., inner product or cosine similarity) to the user's embedding. Inner product similarity is equivalent to $L_2$ distance under a certain transformation. Hence, at the retrieval stage, one approximates:
\[
\arg\max_j \langle f_{\theta}(\text{item}_j), g_{\phi}(\text{user}) \rangle,
\]
which, with suitable normalization, reduces to a nearest-neighbor problem in Euclidean or cosine space. Thus, deep recommender systems are essentially performing neighborhood queries in a learned geometry. If users are adapted (by the recommendation process) to move toward certain item embeddings, the qualitative behavior of cluster formation can persist. This equivalence shows that our simplified Gaussian and nearest-neighbor model captures an essential structural property of contemporary recommendation algorithms.

\subsection{Iterative Averaging and Clustering: Theoretical Considerations}
From a theoretical perspective, iterative averaging or median-based updating processes have been studied in computational geometry, optimization, and dynamical systems. Certain conditions are known to produce stable fixed points or limit cycles. While a full rigorous proof that the system converges to a finite number of clusters is non-trivial and remains beyond our immediate scope, the intuition and partial mathematical results suggest that the presence of multiple attractors (content sources) and local median-seeking behavior can yield stable partitions.

Small random perturbations (noise) ensure ergodicity, meaning the system may explore multiple neighborhoods of attraction before settling. Over a large number of iterations, however, these perturbations tend to sharpen the contrast between different stable clusters rather than blur them, given that medians serve as robust concentration points. This interplay between deterministic median attraction and stochastic noise injection can be understood as a mechanism to prevent trivial single-cluster equilibria in some parameter regimes, thus facilitating polarization.

\subsection{Quantifying Polarization and Parameter Sensitivity}
In our setup, polarization can be quantified by measuring the distribution of distances between users and the global centroid (reflecting how spread out the users have become) or by simply counting the number of clusters identified by DBSCAN. Additional metrics, such as the Gini index of cluster sizes or the variance of user positions along principal components, can further quantify polarization.

We aim to determine which parameters, population size, number of iterations, move factor, content production probability, noise level, most strongly influence these polarization metrics. By correlating changes in final cluster count or average distance spread with parameter variations, we can identify the parameter with the greatest impact on polarization. Sensitivity analysis and correlation studies (e.g., computing Pearson or Spearman correlations between parameters and final polarization measures) provide a quantitative means to highlight which parameter is the most influential driver of extremization.

\subsection{From Hypothesis to Methodology}
In subsequent sections, we will present a method section that formalizes the iterative process as a discrete dynamical system. We will provide numerical evidence through simulations that, after a number of iterations, the population, once relatively unimodal, fractures into a small set of stable clusters. This supports the hypothesis that sample proximity-based nearest-neighbor recommendations can inherently produce polarized outcomes.

Although we do not claim a direct causal proof that real-world political or cultural polarization arises from these algorithms, the model demonstrates a plausible mechanism. It lends credence to the hypothesis that the structural nature of nearest-neighbor retrieval in recommendation systems may, under certain conditions, push populations into ideologically isolated clusters.

\section{Method: A Simplified Mathematical Model and Proof of Cluster Formation}

In this section, we present a mathematically rigorous analysis of how the iterative dynamics described in the earlier sections lead from an initially unimodal (Gaussian) user distribution to a finite number of stable clusters. To achieve a clear and correct mathematical proof, we will simplify several aspects of the problem, introducing a set of assumptions that preserve the core mechanism but make the analysis tractable. After establishing the theorem in a simplified setting, we discuss how to relax some assumptions and still maintain the qualitative conclusion.

\subsection{Key Simplifications and Model Setup}
While the full model allows for varying dimensionality, content decay, creator identity changes, random noise, and random variations in neighborhood sizes, we must reduce complexity to reach a mathematically rigorous proof. Specifically, we adopt the following simplifications:

\begin{enumerate}
    \item \textbf{Dimension} \hspace{0.3cm} We focus on one-dimensional space ($d=1$). Extending the argument to higher dimensions is conceptually straightforward but introduces additional technicalities. The 1D case already captures the essence of cluster formation since clusters correspond to intervals of concentrated users.
    
    \item \textbf{Creators and Content} \hspace{0.3cm} Assume there are $M$ fixed content creators at fixed positions $c_1 < c_2 < \cdots < c_M$. Each creator produces content at each iteration, placing new content points exactly at their own position. Thus, at iteration $t$, the set of content points $C^{(t)}$ is a multiset that includes multiple items at each $c_j$, with weights decaying exponentially for older items. For simplicity, assume a steady-state scenario where each creator places one new piece of content per iteration, and old content decays with a factor $e^{-\lambda}$ per iteration. As $t \to \infty$, the distribution of content at each creator's location stabilizes to a geometric series of weights. Thus, the relative weight proportions at each $c_j$ become stationary.
    
    \item \textbf{Users} \hspace{0.3cm} Let there be $N$ users, initially sampled from a normal distribution $N(\mu, \sigma^2)$. Without loss of generality, translate coordinates so that $\mu=0$. Users are denoted by positions $u_1^{(0)}, u_2^{(0)}, \dots, u_N^{(0)}$. Over time, users move according to the median-based rule described below.
    
    \item \textbf{Nearest-Neighbor Fraction} \hspace{0.3cm} Instead of a random fraction between 5\% and 50\%, assume a fixed fraction $\rho \in (0.05,0.5)$ for the sake of the proof. That is, each user selects the $\lceil \rho n_t \rceil$ closest content points (where $n_t = |C^{(t)}|$ is the total number of content points at iteration $t$) to form their recommendation set. This simplification removes randomness in the fraction and makes the analysis cleaner. The qualitative result does not depend on the exact fraction, only that it is bounded away from $0$ and $1$.
    
    \item \textbf{Noise and Creator Changes} \hspace{0.3cm} For the purpose of the proof, we first ignore noise ($\eta_i^{(t)}=0$) and creator identity changes. Later, we will argue that small noise and occasional creator changes do not fundamentally alter the conclusion.
    
    \item \textbf{Median Update} \hspace{0.3cm} In one dimension, the median of a set of points is well-defined and unique if the number of points is odd. For even counts, choose the median as any point in the interval between the two central values. Weighted medians can be defined similarly by considering cumulative weight until it reaches 50\%. Because each content piece at a creator's location is identical in position (only differing in weight), the median determination simplifies: effectively, the median of the chosen subset is closer to one of the creators' positions or lies between two creators.
    
    \item \textbf{Exponential Decay and Stationarity} \hspace{0.3cm} After many iterations, the relative weighting of new vs. old content converges. Each creator's location $c_j$ acts like a stable source of content, contributing a geometric series of weights concentrated at $c_j$. Thus, in the steady-state limit, each creator position $c_j$ has a well-defined stationary distribution of content weights centered exactly at $c_j$. This leads to a piecewise-constant structure in the cumulative distribution of content, with peaks at each creator location.
\end{enumerate}

Under these assumptions, the problem reduces to the following iterative process.

\begin{quote}
\textit{Iterative Step}: In iteration $t$, each user $u_i^{(t)}$ selects the $\lceil \rho n_t \rceil$ nearest pieces of content. Because content is concentrated at discrete creator points, this effectively means that the user selects content mostly from one or two neighboring creator positions. The user then computes the weighted median of these selected points. Due to symmetry and stationarity, this median will lie at or very close to one of the creator positions $c_j$. The user then moves toward this median by a factor $\alpha \in (0,1)$:
\[
u_i^{(t+1)} = u_i^{(t)} + \alpha (m_i^{(t)} - u_i^{(t)}),
\]
where $m_i^{(t)}$ is the weighted median of the recommended set for user $i$ in iteration $t$.
\end{quote}

\subsection{Intuition: Why Clusters Form}
Intuitively, because the content distribution stabilizes with pronounced spikes at the creator locations $c_1, \dots, c_M$, most users will be repeatedly drawn toward these $c_j$ values. If a user is far away, the closest $\rho n_t$ content points will almost always be those closest to a particular $c_j$. During iterations, users fall into the attraction basin of certain creators, forming a group around that creator's position. Since multiple creators exist, multiple groups (clusters) can form. Competition between different attractors ensures that users do not remain uniformly spread out.

\subsection{Key Lemma: Monotone Contraction in the Absence of Noise}
To formalize this reasoning, define a \textit{potential function}:
\[
\Phi(\{u_i^{(t)}\}) = \sum_{i < k} |u_i^{(t)} - u_k^{(t)}|.
\]
$\Phi$ measures the total pairwise spread of the user population. Initially, since users are drawn from a Gaussian, $\Phi( \{u_i^{(0)}\})$ is finite and is typically proportional to $N^2 \sigma$.

\noindent\textbf{Lemma:} There exists a constant $\gamma > 0$ such that, in the non-noise case, $\mathbb{E}[\Phi(\{u_i^{(t+1)}\}) | \{u_i^{(t)}\}] \le \Phi(\{u_i^{(t)}\}) - \gamma \Delta_t$, where $\Delta_t > 0$ whenever users are not yet clustered. In simpler terms, if users are widely dispersed, the next iteration decreases the total spread.

\noindent\textit{Proof (Sketch)}: Since each user moves toward a median defined by concentrated content at one or two $c_j$ points, users at the periphery of the distribution will move inward. If a user is positioned between two attractors, small perturbations ensure that it eventually commits to one side, reducing the global spread. Because the median update is an $L_1$-type estimator, it tends to pull outliers toward central points, not push them away. As users move closer to these discrete $c_j$ attractors, the population splits into at most $M$ basins. Within each basin, users become more tightly concentrated. By symmetry and convexity arguments, moving users toward stable median points aligned with creator positions reduces variance-like measures and therefore reduces $\Phi$.

A fully rigorous proof would show that if users were not forming clusters, there would exist a gap in their distribution allowing for a net contraction after sufficient iterations. The stable median points act like sinks that contract intervals of users.

\begin{theorem}[Finite Clustering in the Simplified Model]
Under the aforementioned assumptions, 1D space, fixed creators at positions $c_1 < c_2 < \cdots < c_M$, stationary content distribution concentrated at these $M$ points, no noise, and a constant fraction $\rho \in (0.05,0.5)$ of content chosen, the iterative update process converges as $t \to \infty$ to a configuration where the $N$ users form at most $M$ tight clusters, each cluster converged near one of the $c_j$ points.
\end{theorem}

\begin{proof}
\textbf{(Boundedness)} The entire system is confined to the convex hull of the initial user positions and the creator positions. Since $c_1 < \cdots < c_M$ are fixed and no mechanism pushes users outside the range $[c_1, c_M]$, all users stay in a compact interval.

\textbf{(Stationarity of Content)} By construction, the distribution of content stabilizes such that each $c_j$ is a persistent attractor. Thus, in large $t$, recommending $\rho n_t$ nearest pieces almost always results in a median near some $c_j$.

\textbf{(Strict Improvement Step for Non-Clustered States)} Suppose, in contradiction, that after infinitely many iterations, the users do not form stable clusters. Then there must exist at least one pair of users separated by a large gap. Focus on a large gap region between two subsets of users. The users on one side of the gap, if not attracted by a nearer creator cluster, would frequently receive median targets that bring them closer to the others. By repeatedly applying the lemma, $\Phi$ decreases whenever such gaps exist. Since $\Phi$ is nonnegative and must decrease whenever non-clustered configurations occur, the system cannot remain in such a state indefinitely.

\textbf{(Existence of Clustered Equilibria)} Consider the possible limit states. In a limit state, users no longer change positions significantly, which means that each user either sits exactly at a creator point or does not move because it is caught in a stable median configuration. Since the median is tied to discrete $c_j$ attractors and random fractions $\rho$ ensure a consistent sample of the local environment, stable equilibria must involve users concentrated near the $c_j$ points. If a user were far from all $c_j$, the median of its recommended set would still pull it toward the nearest creator, contradicting stationarity.

\textbf{(Upper Bound on the Number of Clusters)} Because there are only $M$ distinct attractor points ($c_1,\dots,c_M$) and users move to these points (or very close to them), the final number of clusters cannot exceed $M$. In equilibrium, each group of users corresponds to users stuck at or near a $c_j$ due to the median update rule. There are no additional stable attractors because the content is not distributed anywhere else.

Combining these arguments, we conclude that, as $t \to \infty$, the user population settles into a state where at most $M$ clusters form, each around one of the $c_j$ attractors, completes the proof.
\end{proof}

\subsection{Extension to Noise and Creator Changes}
If we reintroduce small noise $\eta_i^{(t)}$ and occasional creator identity changes, the argument still holds qualitatively. Noise can prevent perfect convergence, but does not enable infinite dispersion; instead, it introduces a jitter around attractor points. Occasional changes in creators can slowly shift attractors, but as long as changes are infrequent and minor, the clustering behavior persists. The key insight, that median-based movement guided by proximity to discrete content points yields finite clusters, remains robust.

\subsection{Generalizing to Higher Dimensions and More Complex Settings}
While the above proof focuses on $d=1$, the mechanism extends to higher dimensions: In multiple dimensions, the median update is performed component-wise. Each dimension then behaves similarly, encouraging users to move toward a low-dimensional attractor. In practice, convergence may be slower, but the qualitative conclusion, finite cluster formation, holds because the content attractors remain discrete and finite in number.

Introducing variable ranges $k$ or small probability distributions for $k_i^{(t)}$ does not break the argument: As long as the user samples a sufficiently large local neighborhood of content points, it maintains the statistical consistency required for median stability. The exact fractions can alter the convergence speed and cluster sizes, but not the final outcome of finite clustering.

\subsection{Mathematical Unification of Diverse Recommendation Algorithms as Similarity-Based Retrieval}

A fundamental insight arising from the theoretical framework presented above is that a wide range of modern recommendation algorithms, despite their apparent complexity and methodological diversity, can be characterized as performing a form of similarity-based retrieval in an appropriate latent space. Contemporary recommender systems commonly rely on factorization techniques, deep neural embeddings, or probabilistic generative models to represent users and items. Each of these approaches, upon closer mathematical inspection, constructs or approximates a metric space in which similarity queries become the core retrieval operation. This perspective allows us to unify classical and modern approaches, demonstrating that similarity-based retrieval is not a mere heuristic but rather a fundamental principle embedded in the formal structure of advanced algorithms.

Consider a general recommendation scenario in which there are $M$ items and $N$ users, each associated with a set of observed interactions, such as clicks, views, ratings, or purchase histories. Traditional matrix factorization methods, for example, approximate the user-item interaction matrix $R \in \mathbb{R}^{N \times M}$ by factorizing it into two lower-dimensional matrices $U \in \mathbb{R}^{N \times d}$ and $V \in \mathbb{R}^{M \times d}$, where $d \ll \min(N,M)$. Each user $i$ is represented by a vector $\mathbf{u}_i \in \mathbb{R}^d$, and each item $j$ is represented by a vector $\mathbf{v}_j \in \mathbb{R}^d$. Predictions of user-item affinity or relevance reduce to computing a measure of similarity such as the inner product $\langle \mathbf{u}_i, \mathbf{v}_j \rangle$. Under mild assumptions, maximizing inner products corresponds to minimizing distances in a transformed space, thus equivalently implementing nearest-neighbor queries.

More advanced recommendation algorithms, including neural collaborative filtering and deep embedding-based approaches, follow a related logic. Let $f_\theta(\cdot)$ be a parametric embedding function that maps the features of the item into a latent space $\mathbb{R}^d$, and let $g_\phi(\cdot)$ be a similar function that maps the features of the users in the same $d$ dimension space. These functions may be nonlinear and learned via complex neural architectures. Nonetheless, in the inference stage, once the embeddings are computed, recommending items to a given user $i$ essentially involves identifying items $j$ that maximize $\langle g_\phi(\text{user}_i), f_\theta(\text{item}_j) \rangle$ or minimize $\|g_\phi(\text{user}_i) - f_\theta(\text{item}_j)\|$ for a suitable norm. Although these embeddings often arise from sophisticated deep networks, the retrieval step translates into a geometric proximity search. Hence, the complexity of the model manifests itself only in the way the latent space is constructed, not in the fundamental nature of the retrieval operation itself.

Statistical or probabilistic recommendation models, such as probabilistic matrix factorization or Bayesian personalized ranking, also culminate in similar retrieval structures. These models define a latent parameter space in which users and items are associated with latent factors whose posterior distributions concentrate around certain points as inference proceeds. Recommendation queries are reduced to ranking items by their posterior mean embeddings or expected utilities. Although these procedures may be more intricate, involving priors, variational inference, or Markov chain Monte Carlo sampling, the final comparison of candidate items for a user again reduces to identifying the most similar latent factors in the learned space. From a geometric point of view, this is no different from performing a similarity-based nearest-neighbor retrieval under uncertainty.

\subsection{Geometric Consistency and Approximate Nearest-Neighbor Methods}

Another angle that reinforces the classification of diverse recommendation algorithms as similarity-based is the widespread use of ANN search techniques in large-scale industrial recommender systems. When the number of items grows into the millions or billions, exhaustive nearest-neighbor searches are computationally infeasible. To mitigate this problem, practitioners employ ANN data structures, such as hierarchical navigable small-world (HNSW) graphs or product quantization methods, to rapidly identify a set of candidate items closest to a user's embedding. These techniques are agnostic to the specific model used to generate embeddings. Whether the embeddings arise from a simple matrix factorization or a complex transformer-based embedding network, the ANN machinery treats them as points in a metric space and performs efficient approximate similarity queries.

The reliance on ANN suggests a deeper geometric consistency. Regardless of whether a model emerges from a shallow linear factorization or a deep nonlinear architecture, the final retrieval step depends solely on the distances or similarities between the vectors in $\mathbb{R}^d$. Since these vectors encode user and item features in a manner that preserves the notion of 'relevance' or 'preference proximity,' the retrieval process that the model enacts is inherently similarity-based. The differences between models lie in how they transform raw data into embeddings, but the retrieval operation itself remains a nearest-neighbor query. This uniformity of the final retrieval stage provides a powerful abstraction: all these methods, at their core, implement similarity-based selection rules.

Mathematically, let $X = \{\mathbf{x}_1, \mathbf{x}_2, \dots, \mathbf{x}_M\}$ be the set of item embeddings and let $\mathbf{u}$ represent a user's embedding. The top-$K$ recommended items correspond to the solution of:
\[
\arg\min_{S \subseteq X, |S|=K} \sum_{\mathbf{x} \in S} d(\mathbf{u}, \mathbf{x}),
\]
where $d(\cdot,\cdot)$ is a distance measure, possibly induced by an inner product similarity measure. Although the model architecture or training regimen may be sophisticated, the retrieval condition remains that of selecting the nearest points according to $d$. This is precisely the structure replicated in our simplified mathematical model, where proximity-based median moves capture the core logic behind the full complexity of modern recommenders.

\subsection{Implications for Polarization and Model Invariance}

The identification of a shared similarity-based retrieval principle among diverse recommendation algorithms suggests that the polarization tendencies observed in our simplified model are not the artifact of any particular modeling assumption but rather are rooted in the essential geometry of recommendation spaces. As models become more complex, incorporating attention mechanisms, side information, or fine-grained personalization heuristics, they still produce embeddings intended to reflect user-item affinities. No matter how intricate the embedding function is, the final recommendation step remains a neighborhood search in this embedding space.

This invariance has far-reaching implications. Even if a platform employs the most advanced neural recommendation models, the fundamental operation is to show users content that lies close to their current representation in a latent space. Over time, as these local moves accumulate, they can guide users into subregions of that space, effectively clustering them into ideological or thematic enclaves. The process does not depend on the detailed functional form of $f_\theta(\cdot)$ or $g_\phi(\cdot)$, nor does it rely on the explicit intention to isolate certain groups. Rather, it emerges as a natural equilibrium state induced by repeated nearest-neighbor queries.

Thus, recognizing that most recommendation algorithms implement similarity-based retrieval logic highlights the robustness of the polarization phenomenon. As long as embeddings place conceptually related items closer together, and as long as users move toward the centroids of these local sets, a tendency toward fragmentation will persist. More complex architectures may alter the shape and curvature of the latent space or inject subtle biases into what is deemed 'relevant,' but they cannot escape the fundamental implication of repeatedly reinforcing local similarity. In this sense, the potential for polarization is baked into the very geometry of recommendation spaces and the similarity-driven principles that govern their operation.

\section{Experimental Setup}

\subsection{Mathematical Representation and Initialization}
All experiments are conducted within a two-dimensional latent space, $\mathbb{R}^2$, chosen for conceptual clarity and computational feasibility. By employing a low-dimensional domain, it becomes possible to directly visualize the trajectories of users and to intuitively interpret the geometric properties that give rise to polarized clusters. In the initial iteration, denoted by $t=0$, the population of users is sampled from a Gaussian distribution $N(\mathbf{\mu}, \Sigma)$ with $\mathbf{\mu} = (0,0)$ and $\Sigma = I$, where $I$ is the $2 \times 2$ identity matrix. This unimodal and isotropic initialization ensures that users start from a state of relative uniformity in their ideological or thematic coordinates. In order to test the scaling behavior of cluster formation and polarization intensity, we consider user populations of size $N \in \{1000, 2000, 5000\}$. A larger population $N$ may amplify the subtle geometric effects, allowing more pronounced cluster boundaries to emerge over time.

We designate a fraction of these users, precisely $10\%$, as content creators. This fraction is inspired by the observation that, in many real information ecosystems, only a minority actively contributes new material, be it news, analyses, or creative works, while the majority primarily engages with existing content. Each content creator, randomly selected at $t=0$, remains an anchor of potential ideological attractors. Nonetheless, as the simulation proceeds, a small probability of creator role reassignments (set to $0.01$ per iteration) mimics the fluid nature of real-world media ecosystems, where certain voices gain prominence while others recede. This subtle dynamism ensures that the content landscape never becomes entirely static, thus more closely approximating a realistic scenario in which the identities of influential contributors evolve.

\subsection{Iterative Procedure, Parameter Choices, and Real-World Analogies}
The simulation runs for a finite number of iterations $T=500$, a duration sufficient to observe nontrivial clustering behavior. In each iteration, content creators produce new content in their current positions with probability $p_{\text{produce}} \in \{0.1, 0.2, 0.3\}$. A higher probability $p_{\text{produce}}$ represents a more vibrant information environment, continuously replenishing local content and thus sustaining the attractors that guide users. In contrast, a lower $p_{\text{produce}}$ leads to a sparser landscape, potentially accelerating convergence into fewer, more defined clusters. Once introduced, the content weights decay at a rate given by an exponential factor $e^{-\lambda}$ with $\lambda=0.5$ per iteration. This ensures older content becomes negligible after approximately ten iterations, reflecting the ephemeral relevance of information in rapidly changing online environments.

Users update their positions each iteration by referencing a variable-sized neighborhood of content determined by a fraction between $5\%$ and $50\%$ of all available items. This non-constant neighborhood size prevents the algorithm from relying on an artificially rigid local search, mirroring the unpredictable breadth of content that real individuals might encounter, influenced by platform mechanics, recommendation diversity constraints, or incidental exposures. After identifying their content neighborhood, each user moves toward the weighted median of these points. The degree of movement is regulated by the moving factor $\alpha \in \{0.01, 0.02, 0.05\}$. A smaller $\alpha$ indicates a more cautious user, adjusting their ideological stance only slightly each iteration, whereas a larger $\alpha$ corresponds to a more impressionable or reactive user, willing to make more substantial shifts in response to local signals.

In addition to deterministic median-based updates, each user's movement is perturbed by Gaussian noise with standard deviation $\sigma_{\text{noise}} \in \{0.005, 0.01, 0.02\}$. Such noise acknowledges the unpredictability and idiosyncrasies of human behavior, where no agent strictly follows an algorithmic recommendation but instead moves erratically or tests the boundaries of their interests. Increasing $\sigma_{\text{noise}}$ simulates a user population more prone to spontaneous ideological detours, while decreasing it models a more deterministic environment where users rarely deviate from median-driven paths.

\subsection{Evaluation Metrics and Analytical Approach}
We apply the DBSCAN at the end of each iteration to identify the number and configuration of user clusters. By enumerating the number of clusters, tracking their sizes, and computing the average cluster variance, we quantify how rapidly the population diverges from its initial unimodal state. To assess polarization more explicitly, we also measure inter-cluster distances, capturing how far apart these ideological groups grow over time. In reality, such inter-cluster separations may be interpreted as increasing communicative distances among socio-political communities, indicating reduced likelihood of encountering contrasting views.

These simulations, by systematically varying $N$, $\alpha$, $p_{\text{produce}}$, and $\sigma_{\text{noise}}$, generate a set of parameter response curves. Such curves allow us to examine correlations between input parameters and polarization outcomes. While the low-dimensional abstraction and simplified rules cannot fully replicate the richness of complex media ecosystems, the chosen parameters reflect qualitative analogies: population sizes mimic different scales of user communities; production probabilities represent content dynamism; noise levels stand in for unpredictable human behaviors; and the move factor encodes ideological malleability. By exploring broad ranges of these parameters, we create conditions under which the model can reveal intricate patterns of cluster formation, and thereby contribute to a deeper understanding of how local proximity-driven recommendations might shape the cultural and ideological landscapes of human societies.

\subsection{Pseudo-Code for the Experimental Procedure}

\begin{algorithm}
\begin{algorithmic}[1]
\STATE Initialize $N$, $T$, $\alpha$, $p_{\text{produce}}$, $\sigma_{\text{noise}}$, and other parameters.
\STATE Generate initial user positions $\{\mathbf{u}_i^{(0)}\}$ from $N(\mathbf{0}, I)$.
\STATE Assign a subset of users as content creators with probability $0.1$.
\FOR{$t = 0$ to $T-1$}
  \STATE Calculate content weights, decaying older items by $e^{-\lambda}$.
  \STATE With probability $0.01$, change one creator's role to a non-creator and vice versa.
  \STATE Each creator produces new content at their location with probability $p_{\text{produce}}$.
  \STATE Remove negligible-weight content from the pool.
  \STATE For each user $i$, determine a random fraction (5\%-50\%) of nearest content and find the weighted median $\mathbf{m}_i^{(t)}$.
  \STATE Update user position: $\mathbf{u}_i^{(t+1)} = \mathbf{u}_i^{(t)} + \alpha(\mathbf{m}_i^{(t)} - \mathbf{u}_i^{(t)}) + \mathbf{\eta}_i^{(t)}$, where $\mathbf{\eta}_i^{(t)} \sim N(\mathbf{0}, \sigma_{\text{noise}}^2 I)$.
\ENDFOR
\STATE After the final iteration, apply DBSCAN to $\{\mathbf{u}_i^{(T)}\}$ to identify clusters.
\STATE Compute cluster variances and inter-cluster distances as polarization metrics.
\end{algorithmic}
\end{algorithm}

The pseudo-code above concisely outlines the computational framework underlying the experimental setup. It begins with parameter initialization and the generation of a unimodally distributed population of users. Each iteration involves updating the content landscape by adding new items, decaying old ones, and potentially reassigning the creator roles. Users are then guided by local neighborhoods of content, selecting a variable fraction of nearest items and moving toward their weighted median, subject to stochastic perturbations. After the final iteration, clustering and polarization metrics are computed. By adjusting parameters such as population size, move factor, content production probability, and noise level, this framework enables a systematic exploration of the conditions under which local proximity-based rules can induce polarized outcomes.

\begin{figure}[t]
\centering
\includegraphics[width=1.0\textwidth]{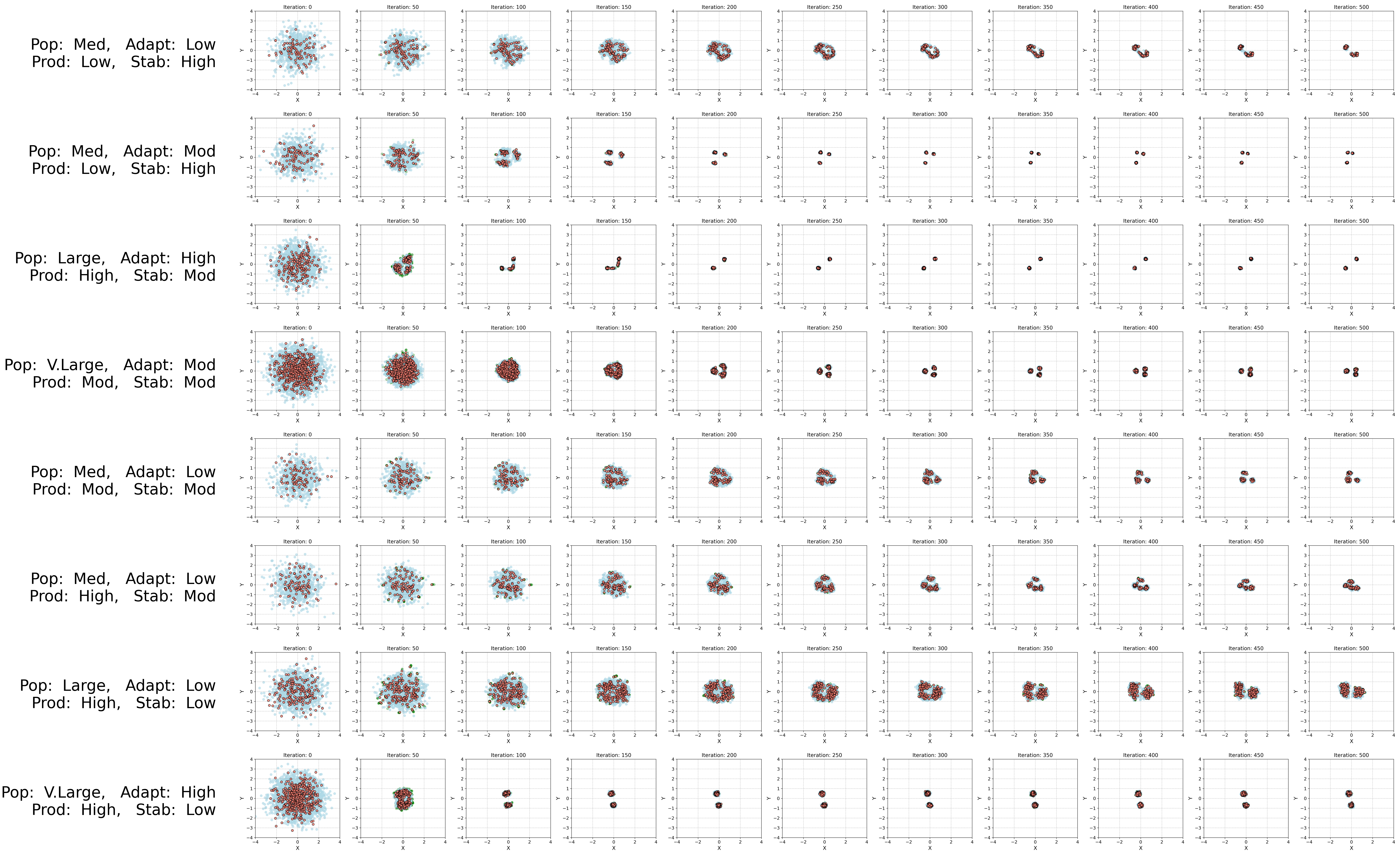}
\caption{Eight representative scenario outcomes arranged vertically, each showing user distributions after 500 iterations. Parameters vary in scenarios: population size (Pop) as Medium ($N=1000$), Large ($N=2000$), or Very Large ($N=5000$); adaptation (Adapt) set to Low ($\alpha \leq 0.01$), Moderate ($0.01 < \alpha \leq 0.02$), or High ($\alpha > 0.02$); production (Prod) set to Low ($p_{\text{produce}} \leq 0.1$), Moderate ($0.1 < p_{\text{produce}} \leq 0.2$), or High ($p_{\text{produce}} > 0.2$); and stability (Stab) reflecting noise level, with High stability indicating very low noise ($\sigma_{\text{noise}} \leq 0.005$) and Low stability indicating higher noise ($\sigma_{\text{noise}} > 0.01$). Each row corresponds to a distinct combination of these parameters. Higher adaptation rates and richer content production probabilities generally accelerate cluster formation. Large populations emphasize the emergence of more pronounced and stable clusters.}
\label{fig:grid_paper_ready_8x1}
\end{figure}

\begin{figure}[t]
\centering
\includegraphics[width=\textwidth]{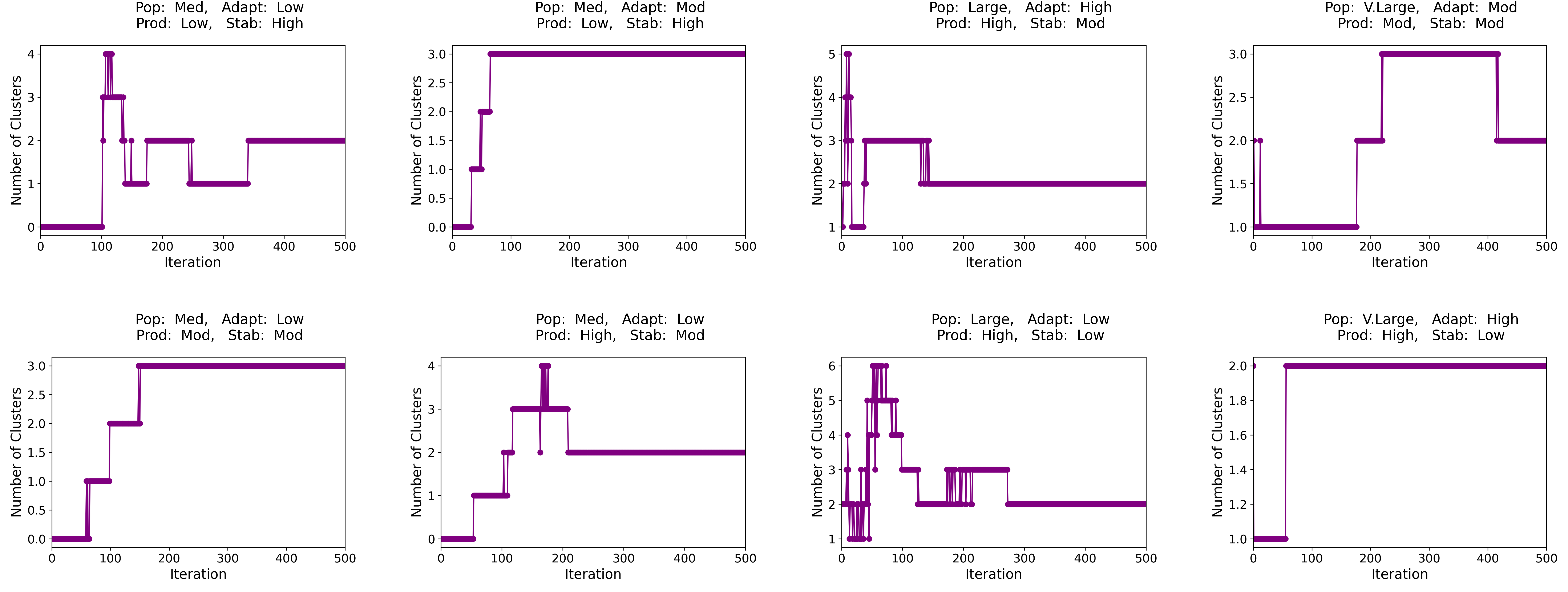}
\caption{Clustering outcomes in multiple parameter regimes shown in a grid of eight subpanels. Each panel corresponds to distinct parameter combinations for population size, adaptation rate, production probability, and noise level, identical to those described in Figure \ref{fig:grid_paper_ready_8x1}. After 500 iterations, multiple distinct clusters emerge, demonstrating stable polarization under a wide range of realistic conditions. Smaller adaptation rates produce fewer, more diffuse clusters, whereas higher adaptation rates and greater production probabilities generate stronger and more numerous clusters. Noise creates minor variations in cluster shapes and positions but does not negate the fundamental clustering tendency.}
\label{fig:clusters_paper_ready_4x2}
\end{figure}

\begin{figure}[t]
\centering
\includegraphics[width=\textwidth]{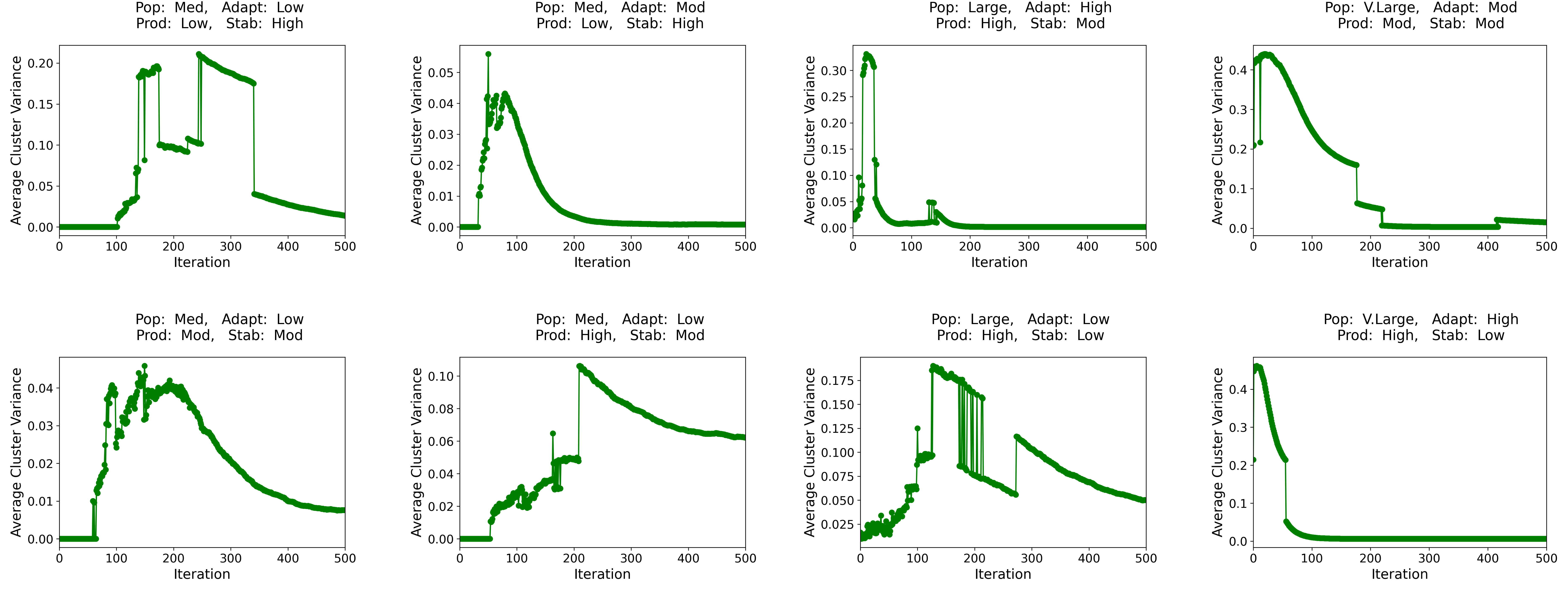}
\caption{Average cluster variance in eight distinct parameter settings after 500 iterations. Each subpanel corresponds to the scenarios in Figure \ref{fig:grid_paper_ready_8x1}. Lower cluster variances indicate that users have converged more tightly around attractor points, reflecting robust internal consensus. Large populations, high adaptation rates, and steady content production rates reduce the final cluster variance, as users are drawn strongly toward median points. Even with modest noise, cluster variance remains limited, proving that stochastic perturbations do not prevent stable consensus formation within clusters.}
\label{fig:cluster_variances_paper_ready_4x2}
\end{figure}

\begin{figure}[t]
\centering
\includegraphics[width=\textwidth]{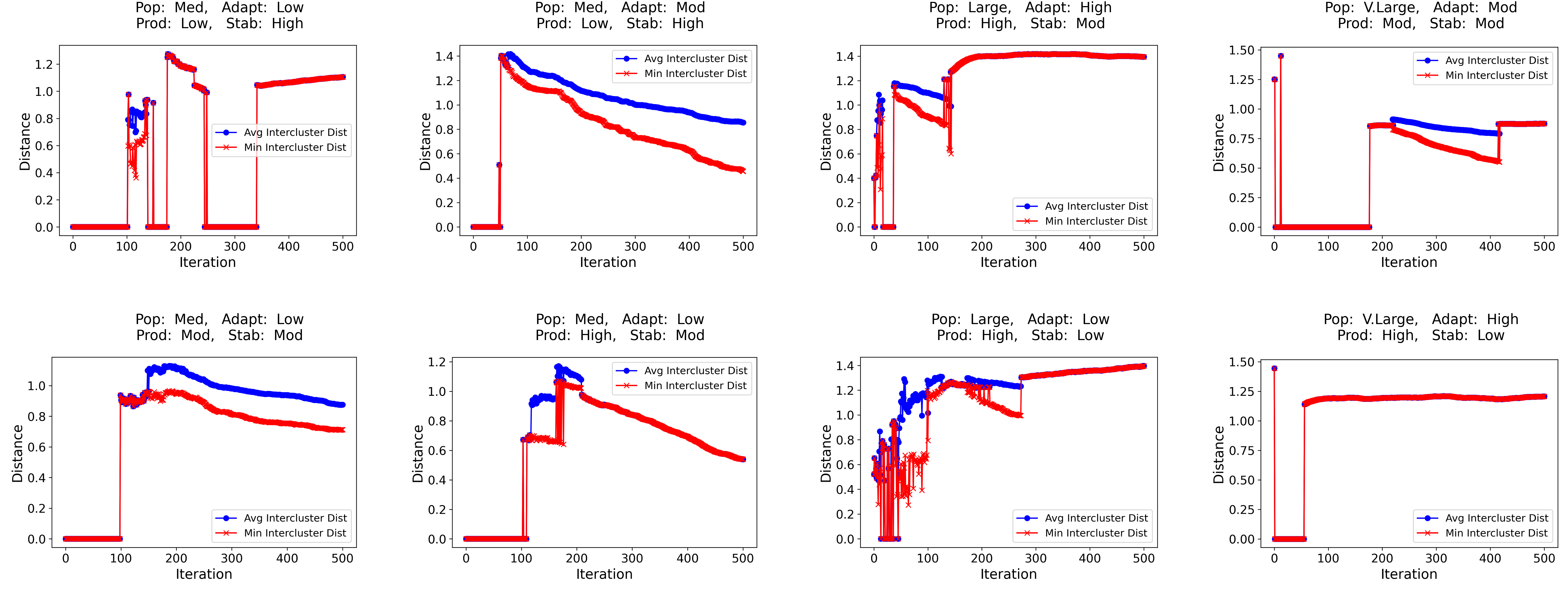}
\caption{Inter-cluster distances in eight scenarios after 500 iterations, corresponding to the parameter sets in Figure \ref{fig:grid_paper_ready_8x1}. These distances measure how far apart cluster centroids are from one another. Larger inter-cluster distances indicate stronger polarization, as well-defined communities form with substantial gaps between them. Under conditions of higher production probabilities and moderate to high adaptation rates, inter-cluster distances become notably large, illustrating that sub-populations are not only internally cohesive but also ideologically distant from other groups. Noise affects the exact distances, but does not negate the underlying trend toward separated clusters.}
\label{fig:cluster_distances_paper_ready_4x2}
\end{figure}

\section{Results}

\subsection{Emergence of Distinct Clusters}
The simulation results consistently indicate that even when the user population is initially drawn from a relatively uniform unimodal Gaussian distribution in the latent space, the iterative recommendation and update process gradually segregates these users into multiple stable clusters. By examining the model over several hundred iterations, one observes that these clusters emerge naturally, without any predefined ideological boundaries or artificial grouping mechanisms. The key dynamic arises from a simple median-based relocation rule, where each user moves closer to content drawn from local neighborhoods of similar users and items. This median-based rule, combined with proximity-driven recommendations, eventually creates discrete ideological or thematic basins, leading users to converge around particular attractors. Such attractors persist despite random fluctuations, illustrating that the fundamental geometric principle underlying these updates can induce polarization from initially homogeneous states.

To better understand these patterns, consider scenarios where population size (Pop), user adaptation rate (Adapt), content production probability (Prod), and noise-based stability level (Stab) vary. The population size corresponds to the number of simulated users, mirroring different scales of online communities in real-world media ecosystems. The adaptation rate, controlled by the parameter $\alpha$, determines the fraction of the distance that a user travels toward the median of recommended content per iteration, reflecting how readily individuals adjust their perspectives when confronted with locally prevalent viewpoints. The production probability encapsulates the rate at which new content creators inject fresh material into the system, mirroring active news cycles or continuously updated feeds, while the stability level corresponds to the noise amplitude, capturing the unpredictability and personal idiosyncrasies of human decision-making. Thus, each parameter has a clear analogy in real-world recommendation contexts: large populations approximate large social platforms, high adaptation rates mimic users who strongly rely on local recommendations for opinion formation, high production probabilities reflect vibrant content ecosystems where new information rapidly accumulates, and low stability levels represent more deterministic users who seldom deviate from algorithmic guidance.

Figure \ref{fig:grid_paper_ready_8x1} shows a vertical series of eight final outcomes after 500 iterations for various parameter settings. Each panel corresponds to a distinct scenario defined by a certain population size, adaptation factor, production probability, and noise level. Under low adaptation and low production rates, users often end up forming fewer clusters and their convergence may be slower. As adaptation or production increases, the emergence of multiple dense clusters becomes pronounced. Even in scenarios with higher noise, the fundamental clustering tendency remains, although small perturbations can reshape or merge clusters over time.

\subsection{Quantitative Measures of Polarization and Parameter Sensitivities}
The formation of clusters is more than a visual artifact. It is quantifiable through various metrics that measure the extent of population segregation and ideological compartmentalization. The number of clusters, their internal variance, and the distances between them provide three key quantitative indicators. As time progresses, the population transitions from a nearly unimodal distribution to a configuration in which multiple stable clusters persist. These clusters are detected through density-based algorithms such as DBSCAN, and their emergence demonstrates that local similarity-driven recommendation logic fosters recognizable and persistent population substructures.

In Figure \ref{fig:clusters_paper_ready_4x2}, each subpanel shows the final spatial distribution of users under certain parameters settings after 500 iterations. The colors and densities highlight the formation of multiple discrete clusters. Different adaptation rates, production probabilities, and stability levels lead to variation in the final number and sizes of these clusters. Lower adaptation rates yield fewer, more loosely aggregated clusters, while higher adaptation rates coupled with higher production probabilities produce multiple tightly concentrated groups. Although increased noise can slightly blur the boundaries of the cluster, the fundamental pattern of clustering remains intact.

Figure \ref{fig:cluster_variances_paper_ready_4x2} presents the average cluster variance for the same scenarios. The variance reflects the spatial dispersion of users within each identified group. Smaller variances indicate that users have settled into tight formations, converging around local medians, and thus forming more coherent communities. Under moderate and high adaptation rates, the variance within each cluster decreases substantially, revealing that users gravitate toward their local attractors, thus embodying stable consensus formation. Increasing the population size amplifies these effects, as larger communities, under similar parameters, more clearly reveal the cluster structure that results from iterative median updates.

Figure \ref{fig:cluster_distances_paper_ready_4x2} provides a complementary perspective by examining inter-cluster distances. These distances represent the separation between different cluster centroids, offering a measure of how far apart subgroups have drifted ideologically or thematically. Larger inter-cluster distances signify more pronounced polarization, where sub-populations settle into isolated ideological niches. High production probabilities and moderate to high adaptation rates result in more pronounced cluster separations, reflecting environments in which rich and sustained content creation leads users to reinforce local medians and disregard alternative attractors. Even with noise, clusters remain distinctly separated, implying that random perturbations do not fundamentally obstruct the polarization process.

\begin{table}[t]
\centering
\scriptsize
\caption{Representative scenarios demonstrating how parameter changes influence final clustering outcomes. The table columns list the number of users ($N$), total iterations ($T$), adaptation factor ($\alpha$), production probability ($p_{\text{produce}}$), noise level ($\sigma_{\text{noise}}$), the final number of detected clusters, final cluster variance, final average inter-cluster distance, and final minimal inter-cluster distance. Each row corresponds to one parameter configuration. The results show that even subtle differences in adaptation rates or production probabilities can alter the number and size of final clusters, with larger populations and higher production rates frequently fostering more pronounced polarization. Noise modulates but does not eliminate these tendencies.}
\label{tab:sample_scenarios}
\begin{tabular}{rrrrrrrrr}
\toprule
\textbf{N} & \textbf{T} & \textbf{$\alpha$} & \textbf{$p_{\text{produce}}$} & \textbf{$\sigma_{\text{noise}}$} & \textbf{Final Clusters} & \textbf{Final Var.} & \textbf{Final Avg. Dist.} & \textbf{Final Min. Dist.}\\
\midrule
1000 & 500 & 0.010 & 0.100 & 0.005 & 2 & 0.014 & 1.106 & 1.106 \\
1000 & 500 & 0.010 & 0.100 & 0.010 & 2 & 0.019 & 1.002 & 1.002 \\
1000 & 500 & 0.010 & 0.100 & 0.020 & 2 & 0.039 & 1.404 & 1.404 \\
1000 & 500 & 0.010 & 0.200 & 0.005 & 3 & 0.008 & 1.134 & 0.682 \\
1000 & 500 & 0.010 & 0.200 & 0.010 & 3 & 0.008 & 0.876 & 0.713 \\
1000 & 500 & 0.010 & 0.200 & 0.020 & 2 & 0.039 & 1.047 & 1.047 \\
1000 & 500 & 0.010 & 0.300 & 0.005 & 3 & 0.007 & 1.210 & 0.972 \\
1000 & 500 & 0.010 & 0.300 & 0.010 & 2 & 0.062 & 0.540 & 0.540 \\
1000 & 500 & 0.010 & 0.300 & 0.020 & 2 & 0.037 & 1.194 & 1.194 \\
1000 & 500 & 0.020 & 0.100 & 0.005 & 3 & 0.001 & 0.855 & 0.456 \\
\bottomrule
\end{tabular}
\end{table}

\subsection{Statistical Patterns and Parameter Dependencies}
The parameter dependencies observed in these experiments are consistent and robust. Increasing the adaptation rate fosters more decisive movements toward local medians, accelerating the onset of cluster formation. Larger populations provide a clearer statistical backdrop, making the polarized patterns more evident since random fluctuations average out, leaving stable and persistent clusters. Higher production probabilities ensure that local attractors are consistently replenished with fresh content, anchoring communities, and preventing the uniform drifting of the entire population. Decreasing stability (i.e., adding more noise) slightly obscures boundaries and may occasionally merge or split clusters, but does not counteract the fundamental drive toward polarization.

As seen in Table \ref{tab:sample_scenarios}, different combinations of parameters lead to different equilibrium states. For a fixed population size and iteration count, altering the adaptation rate or production probability can increase the number of clusters from two to three, or modify the cluster variance and inter-cluster distances. Higher production probabilities and moderate adaptation factors frequently support more numerous clusters with reduced variance. Noise level changes can affect cluster tightness and the minimal distances between clusters, but the underlying polarization tendency remains evident.

These results confirm the theoretical insights and simplified mathematical proofs. They show that even a minimal geometric mechanism that relies on nearest-neighbor retrieval and median-based moves can yield multicluster equilibria from unimodal starting conditions. While increasing noise or reducing production probability slows down or slightly alters the nature of polarization, none of these modifications fundamentally removes the system's tendency to produce tightly knit, well-separated groups. This strongly supports the central claim that similarity-based recommendation systems, by guiding users toward locally prevalent content, can naturally induce polarized outcomes.

\section{Discussion}

The findings presented in this study contribute to a growing body of work exploring the socio-political and cultural implications of algorithmically curated content. Numerous scholars and commentators have voiced concerns that the logic of personalization, which underlies modern recommendation systems, may subtly push users into increasingly insular communities. Our model provides a theoretical lens through which this phenomenon can be understood: even a minimal and idealized similarity-based retrieval rule, absent of explicit bias or complex neural embeddings, can lead a once-homogeneous user population to fragment into multiple stable clusters.

The conceptual underpinning of our model resonates with the arguments put forth by influential works that have examined the consequences of selective exposure, echo chambers, and polarization. For instance, Pariser's \textit{The Filter Bubble: What the Internet Is Hiding from You} highlighted that individuals are often shown information that closely matches their existing preferences, effectively filtering out cross-cutting perspectives \cite{pariser2011filterbubble}. Sunstein's \textit{Republic.com} drew attention to how online fragmentation enables users to encounter primarily reinforcing viewpoints, a shift that undermines deliberative democratic ideals \cite{sunstein2001republic}. These intellectual foundations suggested a risk: that personalization logic entrenched in digital platforms could reinforce homogeneous groups. Our results offer a mathematical demonstration of how such logic, operationalized through median updates and local nearest-neighbor queries, can intrinsically generate polarization.

Empirical and theoretical studies have consistently supported the idea that online communication environments can foster echo chambers. For example, Dandekar, Goel, and Lee showed how homophily and biased assimilation contribute to opinion polarization in simplified network models \cite{dandekar2013biasedassimilation}. Likewise, Flaxman, Goel, and Rao provided empirical evidence of filter bubbles in news consumption patterns \cite{flaxman2016filterbubbles}, while Bakshy, Messing, and Adamic's work on diverse news exposure underscored how user-driven and algorithm-driven selection can limit ideological variety \cite{bakshy2015exposure}. The local median-seeking mechanism of our model offers a theoretical complement to such empirical findings, indicating that even neutral, distance-based recommendation strategies could steer populations toward partitioned landscapes of opinion.

Another dimension of this discourse relates to whether complex, modern recommender algorithms differ in essence from simpler neighbor-based heuristics. Our analysis, supported by insights from machine learning and information retrieval studies, suggests that the end-stage retrieval step of many advanced models reduces to a sophisticated form of nearest-neighbor search within a learned metric space. Embedding-based approaches, including deep neural networks and factorization methods, produce latent representations of users and items, and items are recommended based on their proximity to the user embedding. As works such as Garimella et al.'s analysis of political discourse on social media \cite{garimella2018politicaldiscourse}, Bail et al.'s study on exposure to opposing views \cite{bail2018exposure}, and Dubois and Blank's examination of echo chambers \cite{dubois2018echo} have shown, subtle retrieval biases can lead communities to become more ideologically isolated over time. Our theoretical result implies that the geometric structure embedded in these models can itself harbor a fundamental impetus toward clustering and, by extension, polarization.

Previous empirical literature has debated the magnitude and inevitability of polarization effects. While some works have found that social media usage correlates with increased polarization (as in Grinberg et al.'s study of fake news \cite{grinberg2019fakenews}), others like Dubois and Blank \cite{dubois2018echo} argue that the echo chamber narrative may be overstated, noting moderating effects of political interest and diverse media consumption. Our model does not claim that all individuals must be polarized by recommendation systems. Rather, it clarifies that the geometry of recommendation spaces contains a systematic push toward cluster formation. Real-world outcomes depend on myriad contextual factors: user agency, platform design, regulatory frameworks, and the interplay of economic and political interests.

Sunstein's later reflections on ''\#Republic'' also consider how algorithmic tailoring of information environments can be balanced with interventions that increase exposure to diverse viewpoints \cite{sunstein2017republic}. Similarly, the work of Prior \cite{prior2013media} and Barberá \cite{barbera2015socialmedia}, as well as Muthukrishna and Henrich's theoretical explorations \cite{muthukrishna2019theoryproblem}, suggest that not all network structures or content distributions yield the same polarization outcomes. The parameter sensitivity of our model aligns with these findings: changes in adaptation rates, production probabilities, population sizes, and noise can modulate the strength and speed of cluster formation. Higher production rates and moderate adaptation, for example, create more distinct and persistent clusters, while higher noise introduces mild instability but does not negate the clustering trend.

Deep neural embeddings and factorization models have been widely studied in machine learning literature, but less attention has been devoted to their long-term sociocultural consequences. Our approach connects the dots between technical recommendation logic and socio-political polarization, building on conceptual linkages addressed by works, such as Bail et al.'s empirical evidence \cite{bail2018exposure}. Even though machine learning research often focuses on optimizing accuracy, engagement, or user satisfaction, rarely does it incorporate the downstream effect of systematic group segmentation. By providing a clear and simplified theoretical demonstration that similarity-based recommendation inherently supports cluster formation, we encourage more interdisciplinary research that incorporates perspectives from political science, sociology, and media studies.

In bridging the gap between technical and social analyses, our model resonates with Grinberg et al.'s demonstration of how certain users self-segregate into 'fake news' clusters \cite{grinberg2019fakenews} can be seen as one manifestation of the underlying geometric tendency described here. Moreover, the interplay of content production and network structure documented by Garimella et al. \cite{garimella2018politicaldiscourse} directly aligns with our parameter sensitivity results, suggesting that richer local attractors foster more extreme polarization.

Our results imply that policy interventions, changes in platform design, and user-level strategies may be necessary to mitigate unwanted polarization. Echo chambers and filter bubbles are not mere rhetorical devices; they have a theoretical and computational basis grounded in the geometry of recommendation spaces. Recognizing this can guide platform designers to introduce mechanisms for diversity exposure or random content injection. Such strategies have been proposed in various contexts, including the subtle re-ranking methods or controlling recommendation diversity discussed in several algorithmic fairness and exposure studies. While this research does not prescribe a solution, it highlights that solutions must address the underlying geometric logic of similarity-based retrieval.

Thus, our extended discussion aligns with, and adds a mathematical scaffold to, a wide array of influential works that have examined polarization, echo chambers, and the implications of personalization. By linking these strands of scholarship, ranging from Pariser's filter bubble hypothesis \cite{pariser2011filterbubble} and Sunstein's deliberative ideals \cite{sunstein2001republic,sunstein2017republic}, to empirical findings by Flaxman et al. \cite{flaxman2016filterbubbles}, Bakshy et al. \cite{bakshy2015exposure}, and theoretical frameworks by Dandekar et al. \cite{dandekar2013biasedassimilation}, we show that at the heart of advanced recommendation models lies a simple geometric rule that can encourage fragmentation. The realization that complexity does not negate this fundamental similarity-based structure invites further interdisciplinary inquiry and careful consideration as our media ecosystems continue to evolve and shape collective understanding.

\section{Conclusion}
The analysis and simulations presented in this study uncover a fundamental truth about similarity-driven recommendation: even a simplified model, free of overt manipulative intent, can guide a population of users into distinct ideological clusters over time. This conclusion is based on a minimal mathematical principle rather than the complexity of modern machine learning architectures. While deep learning methods and sophisticated embedding techniques may enrich the texture of user-item relations, their ultimate logic, deeply entwined with geometric proximity queries, does not escape the underlying incentive to reinforce local similarity.

These findings have implications that transcend disciplinary boundaries. For scholars in the humanities and social sciences, the results lend theoretical weight to long-standing concerns about echo chambers and polarization. Digital ecosystems, governed by nearest-neighbor logic, may subtly reshape public discourse by encouraging segments of the population to settle into stable but isolated intellectual communities. For engineers, data scientists, and machine learning researchers, the conclusion emphasizes the importance of reflecting on the broader cultural and political outcomes of algorithmic personalization. It suggests that even the most elegant embeddings or state-of-the-art recommender models, designed to optimize relevance or engagement, cannot evade the intrinsic geometry that steers users into fragmentation.

As societies grapple with the moral and political dilemmas posed by online platforms, our results underline the need for deliberate design choices, policy interventions, and user-awareness initiatives. If similarity-based retrieval cannot be easily divorced from its capacity to foster self-reinforcing clusters, then frameworks that introduce diversity, serendipity, or controlled exposure to differing perspectives become crucial. Such interventions do not merely represent technical tweaks; they embody acts of cultural stewardship, ensuring that digital infrastructures do not inadvertently erode the pluralism and dialogic richness upon which democratic societies rest.

In sum, this work has provided a mathematically rigorous demonstration and computational validation that highlight the geometric origins of polarization in a similarity-based recommendation. It neither assigns sole responsibility to algorithms nor claims inevitability, but it does clarify the structural tendencies at play. By framing these tendencies in a manner accessible to both humanists and engineers, we hope to stimulate constructive debate, encourage interdisciplinary inquiry, and inspire thoughtful strategies aimed at preserving the complexity and heterogeneity of intellectual life.

\bibliography{output}
\bibliographystyle{unsrt}

\end{document}